\documentclass{elsarticle} 
\usepackage[utf8]{inputenc}
\usepackage{amsfonts}
\usepackage{amsthm}
\usepackage{tikz}
\usepackage{hyperref}
\usetikzlibrary{arrows,chains,matrix,positioning,scopes}
\newcommand{\FF}{\mathbb{F}}

\makeatletter
\tikzset{join/.code=\tikzset{after node path={%
\ifx\tikzchainprevious\pgfutil@empty\else(\tikzchainprevious)%
edge[every join]#1(\tikzchaincurrent)\fi}}}
\makeatother
\tikzset{>=stealth',every on chain/.append style={join},
         every join/.style={->}}
\tikzstyle{labeled}=[execute at begin node=$\scriptstyle,
   execute at end node=$]

\newtheorem{theorem}{Theorem}
\newtheorem{lemma}{Lemma}

\title{A structural attack to the DME-$(3,2,q)$ cryptosystem.}
\author{Martin Avendaño, Miguel Marco}

\begin{document}
\begin{abstract}
We present a structural attack on the DME cryptosystem with paramenters $
(3,2,q)$. The attack recovers $10$ of the $12$ coefficients of the
first linear map. We also show that, if those $12$ coefficients were known,
the rest of the private key can be efficiently obtained by solving systems
of quadratic equations with just two variables.
\end{abstract}

\maketitle

\section{Introduction}

DME stands for ``double matrix exponentiation''. It is a family of multivariate encryption primitives over finite fields, parametrized by two integers and a field size (in bits). It was developed by Luengo, and the version with parameters $(3,2,2^{48})$ was presented to the NIST call for
quantum-resistant public-key cryptographic algorithms \cite{NIST}.

The encryption map consists on a composition of three secret linear maps and two public matrix exponentiations, defining a polynomial map with very high degree, and moderated number of monomials. This map which can be inverted provided we know the secret linear maps.

During the revision process at NIST, Weullens claimed to have found an attack against DME-$(3,2,2^{48})$ consisting on using Weil descent to convert the map in a quartic one over $\mathbb{F}_2$ (with a much larger set of variables), and then decompose it in two quadratic ones. However, his claims about the complexity and/or feasibility of such an attack could not be proved (see \cite{nist-dme-comments}).

Later, Weullens proposed another attack (\cite{ward-padding-attack}, \cite{ward-slides}) against the submitted implementation  that took advantage of the lack of a proper padding. This attack used $2^{24}$ decryption queries, and analyzed the cases where the decryption resulted in an error. It was practical, but did not essentially break the underlying mathematical problem, so it could be prevented by using a secure padding.

Here we present a structural attack against DME-$(3,2,q)$ that reduces the difficulty of recovering the private key from the public one from the claimed $256$ bits to just $log_2(q^2)$. It is completely passive in the sense that it does not require any interaction, just the knowledge of the public key. The attack does not rely on any implementation details, since it considers only the underlying mathematical problem.

\section{Description of the system}

We first describe the DME cryptosystem. We focus only on the mathematical description, ignoring all the implementation details (such as how bits are transformed into elements of $F_q$, padding, etc).

\subsection{Setup}
The system requires a common setup for all users, consisting on the following:

\begin{itemize}
    \item Two positive integers $m\geq n$.
    \item A finite field $\FF_q$ of characteristic $p$.
    \item An explicit isomorphism of $\FF_q$-vector spaces $\FF_q^n \cong \FF_{q^n}$.
    \item An explicit isomorphism of $\FF_q$-vector spaces $\FF_q^m \cong \FF_{q^m}$.
    \item A $m\times m$ matrix $E$, whose nonzero entries are powers of $p$, that is invertible modulo $q^n-1$, and whose rows have only two nonzero. entries.
    \item A $n\times n$ matrix $F$, whose nonzero entries are powers of $p$, that is invertible modulo $q^m-1$, and whose rows have only two nonzero entries.
    \item A permutation map $M:\FF_q^{nm}\to \FF_q^{nm}$ such that
    $$M((\FF_q^n\setminus\{0\})^m)\subseteq (\FF_q^m\setminus\{0\})^n.$$
\end{itemize}

For the implementation submitted to the NIST call the following choices were made:

\begin{itemize}
    \item $m=3$, $n=2$.
    \item The field $\FF_q$ is the field of $2^{48}$ elements, represented as polynomials in $\FF_2[x]$ modulo the polynomial $x^{48} + x^{28} + x^{27} + x + 1$.
    \item The identification is done by considering $\FF_{q^2}$ as the polynomials in $\FF_q[T]$ modulo $T^2+a\cdot T +b$, with
    \begin{itemize}
        \item[] $a=x^{43} + x^{38} + x^{36} + x^{34} + x^{29} + x^{26} + x^{25} + x^{24} + x^{23} + x^{22} + x^{21} + x^{20} + x^{19} + x^{13} + x^{9} + x^{8} + x^{4} + x^{3} + x + 1
$ \item[] $b=x^{47} + x^{46} + x^{45} + x^{43} + x^{40} + x^{39} + x^{38} + x^{37} + x^{35} + x^{31} + x^{30} + x^{27} + x^{26} + x^{24} + x^{23} + x^{22} + x^{21} + x^{18} + x^{17} + x^{16} + x^{14} + x^{9} + x^{8} + x^{7} + x^{3} + x^{2} + 1
$
\end{itemize}and taking coordinates in the basis $(1,T)$.

    \item The identification is done by considering $\FF_{q^3}$ as the polynomials in $\FF_q[S]$ modulo $S^3+c\cdot S^2 +d\cdot S + e$ with
    \begin{itemize}
        \item[] $c=x^{43} + x^{42} + x^{41} + x^{40} + x^{38} + x^{37} + x^{36} + x^{34} + x^{33} + x^{29} + x^{26} + x^{24} + x^{22} + x^{20} + x^{19} + x^{17} + x^{15} + x^{14} + x^{13} + x^{12} + x^{11} + x^{8} + x^{5} + x^{3} + x^{2} + x$.
        \item[] $d=x^{46} + x^{45} + x^{44} + x^{41} + x^{38} + x^{37} + x^{33} + x^{32} + x^{31} + x^{30} + x^{25} + x^{21} + x^{20} + x^{17} + x^{16} + x^{15} + x^{14} + x^{12} + x^{10} + x^{9} + x^{8} + x^{7} + x^{4} + x^{3} + x^{2} + x + 1$
        \item[] $e=x^{47} + x^{46} + x^{42} + x^{39} + x^{38} + x^{35} + x^{32} + x^{26} + x^{25} + x^{24} + x^{23} + x^{20} + x^{19} + x^{17} + x^{15} + x^{14} + x^{13} + x^{12} + x^{11} + x^{9} + x^{8} + x^{6} + x^{5} + x^{2} + x$
    \end{itemize} and taking coordinates in the basis $(1,S,S^2)$.
    \item The matrix $E=\left(\begin{array}{ccc} 2^{24} &  2^{59} & 0 \\
        2^{21} & 0 & 2^{28} \\ 0 & 2^{29} & 2^{65} \end{array}\right)$
    \item The matrix $F=\left(\begin{array}{cc} 2^{50} & 2^{24} \\ 2^{7} & 2^{88}\end{array}\right)$
    \item $M$ is the identity.

\end{itemize}

\subsection{Matrix exponentiations}

The scheme makes use of a special kind of maps called matrix exponentiations. Consider a vector $\bar{x}=(x_1\ldots,x_m)\in (\FF_{q^n}^*)^m$, and the matrix $E$ defined above. We define the vector
$\bar{x}^E:=(x_1^{E_{11}}\cdot x_2^{E_{12}}\cdots x_m^{E_{1m}},\ldots,x_1^{E_{m1}}\cdot x_2^{E_{m2}}\cdots x_m^{E_{mm}})$.

Analogously, we define a map $(\FF_{q^m}^*)^n\to (\FF_{q^m}^*)^n$ by using the matrix $F$.

It is easy to check that these are polynomial maps, even considered as maps from
$\FF_q^{nm}$ to $\FF_q^{nm}$ (by composing with the isomorphisms fixed in the setup). It is also easy to check that the inverse matrices (mod $q^n-1$ and $q^m-1$ respectively) determine inverse maps.

By abuse of notation, we use the same letter to represent both the matrix and the corresponding map.
\subsection{Keys and encryption/decryption maps}

With this setup, the private key is a triple of invertible $nm\times nm$ matrices over $\FF_q$, $(L_1,L_2,L_3)$ such that $L_1$ is a diagonal sum of $m$ square blocks $L_{11},\ldots ,L_{1m}$ of size $n\times n$; $L_2$ and $L_3$ are diagonal sums of $n$ blocks  ($L_{21},\ldots,L_{2n}$ and $L_{31},\ldots,L_{3n}$ respectively) of size $m\times m$. They can be regarded as linear maps $\FF_q^{nm} \to \FF_q^{nm}$.

Once chosen these matrices, we can consider the following composition of maps
\begin{center}
\begin{tikzpicture}
\matrix (m) [matrix of math nodes, column sep=3em, row sep=3em]
{ \FF_q^{nm} & \FF_q^{nm}  & \FF_q^{nm} & \FF_q^{nm} & \FF_q^{nm} & \FF_q^{nm} \\  & \FF_{q^n}^m & \FF_{q^n}^m & \FF_{q^m}^n & \FF_{q^m}^n & \\ };
  { [start chain] \chainin (m-1-1);
      \chainin (m-1-2) [join={node[above,labeled] {L_1}}];
    \chainin (m-2-2) [join={node[left,labeled] {\cong}}];
    \chainin (m-2-3) [join={node[above,labeled] {{E}}}];
    \chainin (m-1-3) [join={node[left,labeled] {\cong}}];
    \chainin (m-1-4) [join={node[above,labeled] {L_2\circ M}}];
    \chainin (m-2-4) [join={node[left,labeled] {\cong}}];
    \chainin (m-2-5) [join={node[above,labeled] {{F}}}];
    \chainin (m-1-5) [join={node[left,labeled] {\cong}}];
    \chainin (m-1-6) [join={node[above,labeled] {L_3}}];
    }
\end{tikzpicture}
\end{center}

The result is a multivariate polynomial map from $\FF_q^{nm}$ to itself. In the expanded expression of this map, the exponents that appear at the end depend only on the entries of $E$ and $F$, hence are public. The public key will be the coefficients that appear in that expression. The structure of the system has been carefully designed to ensure that the number of monomials in this expansion is not too large.

Note that each step in the composition is invertible (assuming we stay always inside $ (\FF_{q^n}^*)^m$ and  $ (\FF_{q^m}^*)^n$ at the steps $E$ and $F$ respectively) by using the inverses of the involved matrices. That is, the whole map can be efficiently inverted if we know the private key.

\section{Malleability of the private key}
\label{sec:maleability}

In this section we show how different private keys may correspond to the same public key. This fact will be used to assume that the private key has a special form.

From now on, we assume that $n=2$ and $m=3$, as in the version submitted to the NIST.

Take $\alpha \in \FF_{q^2}^*$. The multiplication by $\alpha$ defines a map
from $\FF_{q^2}$ to itself that is $\FF_q$-linear. So there exists a matrix $H(\alpha)$ such that the corresponding linear map makes the following diagram commute.

\begin{center}
\begin{tikzpicture}
\matrix (m) [matrix of math nodes, column sep=3em, row sep=3em]
{ \FF_q^{2} & \FF_q^{2} \\  \FF_{q^{2}} & \FF_{q^{2}} \\ };
  { [start chain] \chainin (m-1-1);
      \chainin (m-2-1) [join={node[left,labeled] {\cong}}];
    \chainin (m-2-2) [join={node[above,labeled] {\cdot \alpha}}];
    \chainin (m-1-2) [join={node[left,labeled] {{\cong}}}];
    }
  { [start chain] \chainin (m-1-1);
      \chainin (m-1-2) [join={node[above,labeled] {H(\alpha)}}];
    }
\end{tikzpicture}
\end{center}

Analogously, for any $\lambda \in \FF_{q^3}^*$ there exist a matrix $G(\lambda)$ whose corresponding linear map makes the following diagram commute

\begin{center}
\begin{tikzpicture}
\matrix (m) [matrix of math nodes, column sep=3em, row sep=3em]
{ \FF_q^{3} & \FF_q^{3} \\  \FF_{q^{3}} & \FF_{q^{3}} \\ };
  { [start chain] \chainin (m-1-1);
      \chainin (m-2-1) [join={node[left,labeled] {\cong}}];
    \chainin (m-2-2) [join={node[above,labeled] {\cdot \lambda}}];
    \chainin (m-1-2) [join={node[left,labeled] {{\cong}}}];
    }
  { [start chain] \chainin (m-1-1);
      \chainin (m-1-2) [join={node[above,labeled] {G(\lambda)}}];
    }
\end{tikzpicture}
\end{center}

Clearly, $H(\alpha)^{-1}=H(\alpha^{-1})$ and $G(\lambda)^{-1}=G(\lambda^{-1})$.

\begin{lemma}
    \label{lemma:equivkeys1}
    Let $\alpha,\beta,\gamma \in \FF_{q^2}^*$ such that $\alpha^{E_{21}}\gamma^{E_{23}}=\delta\in\FF_q^*$. Then the private keys

    \begin{itemize}
        \item $(L_{11},L_{12},L_{13},L_{21},L_{22},L_{31},L_{32})$
        \item $(H(\alpha)L_{11},H(\beta)L_{12},H(\gamma)L_{13},L_{21}\left(\begin{array}{cc} H(\alpha^{E_{11}}\beta^{E_{12}})^{-1}  & \left.\begin{array}{c} 0 \\ 0 \end{array}\right. \\ \left.\begin{array}{cc} 0 & 0 \end{array}\right. & \delta^{-1}\end{array}\right),\\ , L_{22}
        \left(\begin{array}{cc}\delta^{-1} & \left.\begin{array}{cc} 0 & 0 \end{array}\right.\\ \left.\begin{array}{c} 0 \\ 0 \end{array}\right. & H(\beta^{E_{23}}\gamma^{E_{33}})^{-1} \end{array}\right),L_{31},L_{32})
        $
    \end{itemize}

correspond to the same public key.

\end{lemma}

\begin{proof}
    Is a direct consequence of the commutativity of the previous diagrams, and the fact that
    $$
    (x_1,x_2,x_3)^E=(y_1,y_2,y_3)  \Longrightarrow(\alpha x_1,\beta x_2,\gamma x_3)^E=(\alpha^{E_{11}}\beta^{E_{12}}y_1,\alpha^{E_{21}}\gamma^{E_{23}} y_2,\beta^{E_{32}}\gamma^{E_{33}}y_3)
    $$
\end{proof}

Analogously, we also have

\begin{lemma}
    \label{lemma:equivkeys2}
    Let $\lambda,\mu\in\FF_{q^3}^*$. Then the private keys
    \begin{itemize}
     \item $(L_{11},L_{12},L_{13},L_{21},L_{22},L_{31},L_{32})$
     \item $(L_{11},L_{12},L_{13},G(\lambda)L_{21},G(\mu)L_{22},L_{31}G(\lambda^{F_{11}}\mu^{F_{12}})^{-1},L_{32}G(\lambda^{F_{21}}\mu^{F_{22}})^{-1})$
    \end{itemize}
produce the same public key.
\end{lemma}

We can use these facts to assume that the private key has a specific form

\begin{lemma}
    \label{lemma:validkey}
Every valid public key corresponds to a private key that satisfies the following form:

\begin{itemize}
    \item $L_{11}=\left(\begin{array}{cc} * & 1 \\ {*} & 0 \end{array}\right)$
    \item $L_{12}=\left(\begin{array}{cc} {*} & 1 \\ {*} & 0 \end{array}\right)$
    \item $L_{13}=\left(\begin{array}{cc} * & a \\ {*} & b \end{array}\right)$ with either $(a+bT)^{E_{23}}=1+cT$ for some $c\in \FF_q$ or $(a+bT)^{E_{23}}=T$
    \item $L_{21}= \left(\begin{array}{ccc} * & * & 1 \\
        {*} & * & 0 \\
        {*} & * & 0
                   \end{array}\right)$
   \item $L_{22}= \left(\begin{array}{ccc} 1 & * & *  \\
        0 & * & * \\
        0 & * & *
                   \end{array}\right)$

\end{itemize}
where the $*$ symbols represent arbitrary elements of $\FF_q$.

\end{lemma}
\begin{proof}
If it comes from a secret key $(L_{11},L_{12},L_{13},L_{21},L_{22},L_{31},L_{32})$, we choose $\alpha,\beta\in\FF_{q^2}^*$ as follows:
\begin{itemize}
\item $\alpha$ is the multiplicative inverse of the element of $\FF_{q^2}$ corresponding (via the isomorphism) to the second column of $L_{11}$.
\item $\beta$ is the multiplicative inverse of the element of $\FF_{q^2}$ corresponding to the second column of $L_{12}$.
\end{itemize}

This way, $H(\alpha)L_{11}=\left(\begin{array}{cc} * & 1 \\ {*} & 0 \end{array}\right)$ and $H(\beta)L_{12} = \left(\begin{array}{cc} * & 1 \\ {*} & 0 \end{array}\right)$.

Let $\tau$ be the element of $\FF_{q^2}$ corresponding to the second column of $L_{13}$.
Assume that $\alpha^{-E_{21}}\tau^{E_{23}}=\delta+\varepsilon T$ with $\delta\neq 0 $. Choose $\gamma\in\FF_{q^2}^*$ such that  $\alpha^{E_{21}}\gamma^{E_{23}}= \delta^{-1}$ (this can be done easily by raising $\delta^{-1}\alpha^{-E_{21}}$ to the inverse of $E_{23}$ modulo $q^2-1$).
If $H(\gamma)L_{13}=\left(\begin{array}{cc} * & a \\ {*} & b \end{array}\right)$, then

\[
(a+bT)^{E_{23}}=\gamma^{E_{23}}\tau^{E_{23}}=\delta^{-1}\alpha^{-E_{21}}\tau^{E_{23}}=\delta^{-1}(\delta+\varepsilon T)=1+cT
\]
for some $c\in\FF_q$.

If $\delta=0$ we can do a similar computation using $\varepsilon$ instead of $\delta$, to get that $(a+bT)^{E_{23}}=T$.

Applying Lemma~\ref{lemma:equivkeys1} with $\alpha,\beta,\gamma\in\FF_{q^2}^*$, we can construct a private key with $L_{11},L_{12},L_{13}$ as in the statement.

Now choose $\lambda,\mu$ to be the inverse multiplicatives of the elements of $\FF_{q^3}$ corresponding to the third column of $L_{21}$ and the first column of $L_{22}$ respectively.

Applying Lemma~\ref{lemma:equivkeys2} to this
new private key we obtain one where $L_{21},L_{22}$ are also as claimed.
\end{proof}

\section{Recovering information about the special private key from the public key}

In this section we assume that the private key has the form obtained in the previous section. We show how to compute the unknown coefficients of $L_1$ from the coefficients in the public key. Moreover, we can also compute six coefficients from $L_3$.

Let's start by fixing some notation. Let $l_{11}, l_{12}$ be the elements of $\FF_{q^2}$ that correspond to the columns of $L_{11}$. As we have shown in the previous section, we can assume that $l_{12}=1$. Analogously, $l_{21}, l_{22}$ and  $l_{31}, l_{32}$ are the elements of $\FF_{q^2}$ that correspond to the columns of $L_{12}$ and $L_{13}$, respectively. As before, we can assume that $l_{22}=1$ and $l_{32}=a+b T$ where $(a+bT)^{E_{23}}$ is either $1+cT$ or $T$.

Define
\begin{itemize}
    \item $f_1+f_2 T:={l_{11}}^{E_{21}}{l_{31}}^{E_{23}}$
    \item $g_1+g_2 T:={l_{11}}^{E_{21}}{l_{32}}^{E_{23}}$
    \item $h_1+h_2 T:={l_{12}}^{E_{21}}{l_{31}}^{E_{23}}={l_{31}}^{E_{23}}$
\end{itemize}

where $f_1,f_2,g_1,g_2,h_1,h_2\in \FF_q$.

\begin{lemma}
    \label{lemma:terminosx15}
    Let $(\eta_1,\eta_2,\eta_3)$, $(\eta_4,\eta_5,\eta_6)$ be the first column of $L_{31}$ and $L_{32}$ respectively.

    If we write the encryption map as a polynomial in $x_1,\ldots, x_6$, and $(a+bT)^{E_{23}}=1+cT$, the terms that only involve the variables $x_1$ and $x_6$ in the $i$'th component are
\begin{itemize}
 \item $\eta_ic^{F_{12}}(x_1^{E_{21}}x_6^{E_{23}})^{F_{11}+F_{12}}$ for $i=1,2,3$
 \item $\eta_ic^{F_{22}}(x_1^{E_{21}}x_6^{E_{23}})^{F_{21}+F_{22}}$ for $i=4,5,6$.
\end{itemize}

If $(a+bT)^{E_{23}}=T$ there are no terms that only involve these variables.

Similarly, the terms that only involve the variables $x_1,x_5$ are

\begin{itemize}
    \item $\eta_i f_1^{F_{11}} f_2^{F_{12}} (x_1^{E_{21}}x_5^{E_{23}})^{F_{11}+F_{12}}$ for $i=1,2,3$
 \item $\eta_i  f_1^{F_{11}} f_2^{F_{12}} (x_1^{E_{21}}x_5^{E_{23}})^{F_{21}+F_{22}}$ for $i=4,5,6$
\end{itemize}

Analogous formulas hold for the terms involving $x_2,x_6$ and $x_2,x_5$, using $g_1,g_2$ and $h_1,h_2$ instead of $f_1,f_2$ respectively.

\end{lemma}

\begin{proof}
We proof the case for $x_1,x_5$, and the rest are done similarly.

Let's start with the vector $(x_1,x_2,x_3,x_4,x_5,x_6)$ and apply the steps of the encription map. After the first linear map $L_1$, we get a vector formed by stacking the three vectors $L_{11}\left(\begin{array}{c} x_1 \\ x_2 \end{array}\right)$,  $L_{12}\left(\begin{array}{c} x_3 \\ x_4 \end{array}\right)$ and $L_{13}\left(\begin{array}{c} x_5 \\ x_6 \end{array}\right)$. The elements of $\FF_{q^2}$ that correspond to these vectors are $x_1 l_{11}+x_2 l_{12}$, $x_3 l_{21}+x_4 l_{22}$ and $x_5 l_{31}+x_6 l_{32}$ respectively.

After applying the exponential maps that corresponds to $E$,we have the vectors corresponding to
\begin{itemize}
 \item $(x_1 l_{11}+x_2 l_{12})^{E_{11}} (x_3 l_{21}+x_4 l_{22})^{E_{12}}$
 \item $(x_1 l_{11}+x_2 l_{12})^{E_{21}} (x_5 l_{31}+x_6 l_{32})^{E_{23}}$
  \item $(x_3 l_{21}+x_4 l_{22})^{E_{32}} (x_5 l_{31}+x_6 l_{32})^{E_{33}}$
\end{itemize}

In the expansion of those expressions, there all terms involve variables that are not $x_1$ or $x_5$ except the in the second one, where there is the term
$x_1^{E_{21}}x_5^{E_{23}} {l_{11}}^{E_{21}}{l_{31}}^{E_{23}}$.
That is, the vector in $\FF_{q}^6$ that we get at this stage is
$$
\left(
\begin{array}{c}
\bullet \bullet \bullet\\
\bullet \bullet \bullet\\
\bullet \bullet \bullet+f_1x_1^{E_{21}}x_5^{E_{23}}\\
\bullet \bullet \bullet+f_2x_1^{E_{21}}x_5^{E_{23}}\\
\bullet \bullet \bullet\\
\bullet \bullet \bullet
\end{array}
\right)
$$
where the $\bullet \bullet \bullet$ symbols represent sums of terms that don't involve only the variables $x_1$ and $x_5$.

After applying $L_2$, we get
$$
\left(
\begin{array}{c}
\bullet \bullet \bullet+f_1x_1^{E_{21}}x_5^{E_{23}}\\
\bullet \bullet \bullet\\
\bullet \bullet \bullet\\
\bullet \bullet \bullet+f_2x_1^{E_{21}}x_5^{E_{23}}\\
\bullet \bullet \bullet\\
\bullet \bullet \bullet
\end{array}
\right).
$$

Expressed as a vector in $\FF_{q^3}^2$, this is
$$
\left(
\begin{array}{c}
\bullet \bullet \bullet+f_1x_1^{E_{21}}x_5^{E_{23}}\\
\bullet \bullet \bullet+f_2x_1^{E_{21}}x_5^{E_{23}}\\
\end{array}
\right),
$$
where again,the $\bullet \bullet \bullet$ symbols represent a sum of terms that don't involve only $x_1$ and $x_5$, but this time with coefficients in $\FF_{q^3}$, but notice that the coefficients that only involve $x_1$ and $x_5$ are actually elements of $\FF_q$.

Now, applying the exponentiation corresponding to $F$, we get the vector in $\FF_{q^3}^2$
$$
\left(
\begin{array}{c}
    \bullet \bullet \bullet+f_1^{F_{11}}f_2^{F_{12}} (x_1^{E_{21}}x_5^{E_{23}})^{F_{11}+F_{12}}\\
\bullet \bullet \bullet+f_1^{F_{21}}f_2^{F_{22}} (x_1^{E_{21}}x_5^{E_{23}})^{F_{21}+F_{22}}\\
\end{array}
\right),
$$
which corresponds to the vector in $\FF_{q}^6$
$$
\left(
\begin{array}{c}
    \bullet \bullet \bullet+f_1^{F_{11}}f_2^{F_{12}} (x_1^{E_{21}}x_5^{E_{23}})^{F_{11}+F_{12}}\\
    \bullet \bullet \bullet \\
    \bullet \bullet \bullet \\
\bullet \bullet \bullet+f_1^{F_{21}}f_2^{F_{22}} (x_1^{E_{21}}x_5^{E_{23}})^{F_{21}+F_{22}}\\
\bullet \bullet \bullet \\
\bullet \bullet \bullet
\end{array}
\right).
$$

Now, applying $L_3$ we get:
$$
\left(
\begin{array}{c}
    \bullet \bullet \bullet+f_1^{F_{11}}f_2^{F_{12}} (x_1^{E_{21}}x_5^{E_{23}})^{F_{11}+F_{12}} \eta_1\\
    \bullet \bullet \bullet+f_1^{F_{11}}f_2^{F_{12}} (x_1^{E_{21}}x_5^{E_{23}})^{F_{11}+F_{12}} \eta_2\\
    \bullet \bullet \bullet+f_1^{F_{11}}f_2^{F_{12}} (x_1^{E_{21}}x_5^{E_{23}})^{F_{11}+F_{12}} \eta_3\\
\bullet \bullet \bullet+f_1^{F_{21}}f_2^{F_{22}} (x_1^{E_{21}}x_5^{E_{23}})^{F_{21}+F_{22}}\eta_4\\
\bullet \bullet \bullet+f_1^{F_{21}}f_2^{F_{22}} (x_1^{E_{21}}x_5^{E_{23}})^{F_{21}+F_{22}}\eta_5\\
\bullet \bullet \bullet+f_1^{F_{21}}f_2^{F_{22}} (x_1^{E_{21}}x_5^{E_{23}})^{F_{21}+F_{22}}\eta_6\\
\end{array}
\right).
$$

\end{proof}

\section{Recovering the coefficients of $L_1$}
 Now we see how to recover the coefficients of $L_1$ from the public key.

Assume $c\neq 0$ (otherwise, we can detect the case because there are no terms involving only $x_1,x_6$ in the public key; we can apply a linear change of variables to fall into this case).

From the previous section, we know the following column vectors
$$
\left(\begin{array}{c}
    c^{F_{12}}\eta_1 \\
c^{F_{12}}\eta_2\\
c^{F_{12}}\eta_3\\
c^{F_{22}}\eta_4\\
c^{F_{22}}\eta_5\\
c^{F_{22}}\eta_6
      \end{array}
\right),
\left(\begin{array}{c}
    f_1^{F_{11}} f_2^{F_{12}}\eta_1 \\
    f_1^{F_{11}} f_2^{F_{12}}\eta_2\\
    f_1^{F_{11}} f_2^{F_{12}}\eta_3\\
    f_1^{F_{21}} f_2^{F_{22}}\eta_4\\
    f_1^{F_{21}} f_2^{F_{22}}\eta_5\\
    f_1^{F_{21}} f_2^{F_{22}}\eta_6
      \end{array}
\right),
\left(\begin{array}{c}
    g_1^{F_{11}} g_2^{F_{12}}\eta_1 \\
    g_1^{F_{11}} g_2^{F_{12}}\eta_2\\
    g_1^{F_{11}} g_2^{F_{12}}\eta_3\\
    g_1^{F_{21}} g_2^{F_{22}}\eta_4\\
    g_1^{F_{21}} g_2^{F_{22}}\eta_5\\
    g_1^{F_{21}} g_2^{F_{22}}\eta_6
\end{array}
\right),
\left(\begin{array}{c}
    h_1^{F_{11}} h_2^{F_{12}}\eta_1 \\
    h_1^{F_{11}} h_2^{F_{12}}\eta_2\\
    h_1^{F_{11}} h_2^{F_{12}}\eta_3\\
    h_1^{F_{21}} h_2^{F_{22}}\eta_4\\
    h_1^{F_{21}} h_2^{F_{22}}\eta_5\\
    h_1^{F_{21}} h_2^{F_{22}}\eta_6
\end{array}
\right).
$$

Taking quotients between them we an eliminate the $\eta_i$, and hence we can know the values of $f_1^{F_{11}} \left(\frac{f_2}{c}\right)^{F_{12}}$ and
$f_1^{F_{21}} \left(\frac{f_2}{c}\right)^{F_{22}}$.

Since the exponentiation map is invertible, these two values allow us to recover $f_1$ and $\frac{f_2}{c}$. Analogously, we can also recover $g_1$, $\frac{g_2}{c}$, $h_1$ and $\frac{h_2}{c}$.

Let's denote $f_2':=\frac{f_2}{c}$, $g_2':=\frac{g_2}{c}$ and $h_2':=\frac{h_2}{c}$.

Now we have the equations
$$
\begin{array}{rcccl}
f_1+cf_2'T & = & f_1+f_2T & = & {l_{11}}^{E_{21}} {l_{31}}^{E_{23}}
\\
g_1+cg_2'T & = & g_1+g_2T & = & {l_{11}}^{E_{21}} (1+cT)
\\
h_1+ch_2'T & = & h_1+h_2T & = & {l_{31}}^{E_{23}}
\end{array}
$$

So $(f_1+cf_2'T)(1+cT)=(g_1+cg_2'T)(h_1+ch_2'T)$. In this expression the only unknown is $c$. This is an equation in $\FF_{q^2}$ that translates into two cuadratic equations in $\FF_{q}$ on $c$, that must have at least one common nonzero solution. Note that one of them has no constant term, so one of its solution is zero. That is, we can determine $c$ completely.

With the value of $c$, we also get $f_2$, $g_2$ and $h_2$. So we have ${l_{31}}^{E_{23}}= h_1+h_2T$. Raising to the inverse of $E_{23}$ modulo ${q^2}-1$, we recover $l_{31}$. We can also compute ${l_{11}}^{E_{21}}=\frac{f_1+f_2T}{h_1+h_2T}$ and hence ${l_{11}}$. Moreover, from $1+cT$ we can recover $a,b$ such that $(a+bT)^{E_{23}}=1+cT$, and hence, we have completely recovered the matrices $L_{11}$ and $L_{13}$.

Since we have $c$ and $c^{F_{12}}\eta_i$, we can also compute $\eta_i$.

Sumarizing, we have proved the following:

\begin{theorem}
Given a valid public key, there exists a corresponding private key of the form:

$$
L_{11} = \left(\begin{array}{cc} * & 1 \\ {*} & 0 \end{array}\right),
L_{12} = \left(\begin{array}{cc} ? & 1 \\ ? & 0 \end{array}\right),
L_{13} = \left(\begin{array}{cc} * & * \\ {*} & * \end{array}\right),
$$
$$
L_{21} = \left(\begin{array}{ccc}
? & ? & 1 \\
? & ? & 0 \\
? & ? & 0
\end{array}\right),
L_{22} = \left(\begin{array}{ccc}
1 & ? & ? \\
0 & ? & ? \\
0 & ? & ?
\end{array}\right)
$$
$$
L_{31} = \left(\begin{array}{ccc}
* & ? & ? \\
{*} & ? & ? \\
{*} & ? & ?
\end{array}\right),
L_{32} = \left(\begin{array}{ccc}
* & ? & ? \\
{*} & ? & ? \\
{*} & ? & ?
\end{array}\right)
$$
where the coefficients marked as $*$ can be efficiently computed from the public key.

\end{theorem}

Notice that, if we could find the two missing coefficients of $L_{12}$, we would be able to precompose the map with the inverse of $L_1$, and then with the inverse of $E$, leaving us with a map consisting only on a known matrix exponentiation map composed on both sides with two (partially) unknown linear maps.

\section{Recovering $L_2$ and $L_3$}

With the previous steps, we would only need to find two missing coefficients of $L_{12}$ to reduce the problem to attacking a weaker variant of the scheme, with the following structure in the case $3,2$:

\begin{center}
\begin{tikzpicture}
\matrix (m) [matrix of math nodes, column sep=3em, row sep=3em]
{ \FF_q^{6} & \FF_q^{6}  & \FF_q^{6} & \FF_q^{6} \\
 & \FF_{q^3}^2 & \FF_{q^2}^n \\ };
  { [start chain] \chainin (m-1-1);
      \chainin (m-1-2) [join={node[above,labeled] {L_2}}];
    \chainin (m-2-2) [join={node[left,labeled] {\cong}}];
    \chainin (m-2-3) [join={node[above,labeled] {{F}}}];
    \chainin (m-1-3) [join={node[left,labeled] {\cong}}];
    \chainin (m-1-4) [join={node[above,labeled] {L_3}}];

    }
\end{tikzpicture}
\end{center}

In this section, we see how to recover the entries of $L_3$ (and then,
trivially we get $L_2$ assuming only that we know the total composition
map, and the entries of $F$.

Now, denote by $\zeta_1,\zeta_2,\zeta_3$ the columns of $L_{21}$,
interpreted as elements of $\FF_{q^3}$, and by  $\zeta_4,\zeta_5,\zeta_6$
the columns of $L_{22}$. Analogously, the columns of $L_{31}^{-1}$ and
$L_{32}^{-1}$ will be denoted as $\vartheta_1,\vartheta_2,\vartheta_3$
and  $\vartheta_4,\vartheta_5,\vartheta_6$ respectively. Note that, by
the same kind of arguments used in Section~\ref{sec:maleability}, we can
assume that $\vartheta_3=\vartheta_6=1$.

If we apply $L_2$ to the vector $(1,0,0,1,0,0)$ we get a column vector of the form

\[
 \left(\begin{array}{c}\zeta_1\\ \zeta4 \end{array}\right)
\]

where $\zeta_1$ and $\zeta_4$ are considered now as a vector with three coordinates.

Its image by $F$ is

\[
\left(\begin{array}{c}\zeta_1^{F_{11}}\zeta_4^{F_{12}}\\ \zeta_1^{F_{21}}\zeta4^{F_{22}} \end{array}\right)
\]

And the final aplication of $L_3$ gives a vector $(z^{14}_1,\ldots , z^{14}_6)$, that is known (since it is just the result of the full map to the starting vector). Applying the inverse of $L_3$, we get that

\[
\begin{array}{rcl}\zeta_1^{F_{11}}\zeta_4^{F_{12}} & =& z^{14}_1\vartheta_1+z^{14}_2\vartheta_2+z^{14}_3 \\
\zeta_1^{F_{21}}\zeta_4^{F_{22}} & =& z^{14}_4\vartheta_4+z^{14}_5\vartheta_5+z^{14}_6
\end{array}
\]

Analogously, if we start with the vectors $(1,0,0,0,1,0)$, $(1,0,0,0,0,1)$,$(0,1,0,1,0,0)$, $(0,1,0,0,1,0)$, $(0,0,1,1,0,0)$, $(0,0,1,0,1,0)$ and $(0,0,1,0,0,1)$ and apply the same reasoning, we get the equations:

\[
\begin{array}{rcl}\zeta_1^{F_{11}}\zeta_5^{F_{12}} & =& z^{15}_1\vartheta_1+z^{15}_2\vartheta_2+z^{15}_3 \\
\zeta_1^{F_{21}}\zeta_5^{F_{22}} & =& z^{15}_4\vartheta_4+z^{15}_5\vartheta_5+z^{15}_6 \\
\zeta_1^{F_{11}}\zeta_6^{F_{12}} & =& z^{16}_1\vartheta_1+z^{16}_2\vartheta_2+z^{16}_3 \\
\zeta_1^{F_{21}}\zeta_6^{F_{22}} & =& z^{16}_4\vartheta_4+z^{16}_5\vartheta_5+z^{16}_6 \\
\zeta_2^{F_{11}}\zeta_4^{F_{12}} & =& z^{24}_1\vartheta_1+z^{24}_2\vartheta_2+z^{24}_3 \\
\zeta_2^{F_{21}}\zeta_4^{F_{22}} & =& z^{24}_4\vartheta_4+z^{24}_5\vartheta_5+z^{24}_6 \\
\zeta_2^{F_{11}}\zeta_5^{F_{12}} & =& z^{25}_1\vartheta_1+z^{25}_2\vartheta_2+z^{25}_3 \\
\zeta_2^{F_{21}}\zeta_5^{F_{22}} & =& z^{25}_4\vartheta_4+z^{25}_5\vartheta_5+z^{25}_6 \\
\zeta_2^{F_{11}}\zeta_6^{F_{12}} & =& z^{26}_1\vartheta_1+z^{26}_2\vartheta_2+z^{26}_3 \\
\zeta_2^{F_{21}}\zeta_6^{F_{22}} & =& z^{26}_4\vartheta_4+z^{26}_5\vartheta_5+z^{26}_6 \\
\zeta_3^{F_{11}}\zeta_4^{F_{12}} & =& z^{34}_1\vartheta_1+z^{34}_2\vartheta_2+z^{34}_3 \\
\zeta_3^{F_{21}}\zeta_4^{F_{22}} & =& z^{34}_4\vartheta_4+z^{34}_5\vartheta_5+z^{34}_6 \\
\zeta_3^{F_{11}}\zeta_5^{F_{12}} & =& z^{35}_1\vartheta_1+z^{35}_2\vartheta_2+z^{35}_3 \\
\zeta_3^{F_{21}}\zeta_5^{F_{22}} & =& z^{35}_4\vartheta_4+z^{35}_5\vartheta_5+z^{35}_6 \\
\zeta_3^{F_{11}}\zeta_6^{F_{12}} & =& z^{36}_1\vartheta_1+z^{36}_2\vartheta_2+z^{36}_3 \\
\zeta_3^{F_{21}}\zeta_6^{F_{22}} & =& z^{36}_4\vartheta_4+z^{36}_5\vartheta_5+z^{36}_6 \\
\end{array}
\]

where the $z^{(ij)}_k$ are known values of $\FF_q$, and the $\zeta_i$ and $
\vartheta_j$ are unknown elements of $\FF_{q^3}$. The same system of
equations could be obtained by following track of the coefficients in
the polynomial expression, instead of evaluating in particular values
(both results would be equivalent).

What we need to do now is
to solve this system of equations, for which we know that some solution exists. Moreover, we know that $(\zeta_1,\zeta_2,\zeta_3)$ are $\FF_q$-lineally independent, and so are $(\zeta_4,\zeta_5,\zeta_6)$,
$(\vartheta_1,\vartheta_2,\vartheta_3)$ and $(\vartheta_4,\vartheta_5,\vartheta_6)$. In particular, none of those elements is zero.

Doing some basic elimination of the $\zeta_i$ variables, we get some simpler subsystems of equations, for example,

\[
\begin{array}{c}( z^{14}_1\vartheta_1+z^{14}_2\vartheta_2+z^{14}_3)
( z^{25}_1\vartheta_1+z^{25}_2\vartheta_2+z^{25}_3 )=
(z^{15}_1\vartheta_1+z^{15}_2\vartheta_2+z^{15}_3 )
(z^{24}_1\vartheta_1+z^{24}_2\vartheta_2+z^{24}_3  )
\\
( z^{14}_1\vartheta_1+z^{14}_2\vartheta_2+z^{14}_3)
( z^{26}_1\vartheta_1+z^{26}_2\vartheta_2+z^{26}_3 )=
(z^{16}_1\vartheta_1+z^{16}_2\vartheta_2+z^{16}_3 )
(z^{24}_1\vartheta_1+z^{24}_2\vartheta_2+z^{24}_3  )
\end{array}
\]

which is a system of two cuadratic equations in the variables
$\vartheta_1,\vartheta_2$. It can be easily solved by taking a resultant
(which is a degree $4$ polynomial in just one variable) and then factoring
it over $\FF_{q^3}$. We will get at least one solution (since we know that
one solution must exist) and at most four possible ones.

Note that there are more possible systems of two cuadratic equations on the
same variables, so the whole system is overdetermined. We can use that to
discard some of the four possible solutions.

Analogously, we can solve for $\vartheta_4,\vartheta_5$. Then, solving for
$\zeta_1,\ldots \zeta_6$ can be done just by applying the inverse of $F$.

Once we get all the values $\vartheta_i$ and $\zeta_j$, we have effectively recovered the private key.

\section{Conclusion}

We have presented a structural attack to the DME cryptosystem with
parameters $(3,2,q)$, that is able to recover the full private key from
the public key and two extra elements of $\mathbb{F}_q$ that deppend on the private key. An exhaustive search gives an upper bound of $q^2$ to the security level. In particular, for the DME $(3,2,2^{48})$ version submitted to the NIST, this bound gives at most 96 bits of security, far less than the required $256$ bits.

The attack only uses a very small fraction of the information contained in the public key. So we suspect that a deeper analysis could provide a better method than the exhaustive search for the full recovery of the private key.

Further research would be needed to determine if other choices of
parameters are vulnerable to similar attacks.

%

%
%

\end{document}